\documentclass[reqno,12pt]{amsart}
\numberwithin{equation}{section}
\usepackage{amsfonts, amsthm}
\newtheorem{lemma}{Lemma}
\newtheorem{theorem}{Theorem}

\theoremstyle{definition}

\selectfont 
\usepackage{hyperref}
\usepackage{orcidlink}

\begin{document}

\title[Transformations preserving the Darboux integrability]{On B\"acklund transformations preserving the Darboux integrability of hyperbolic equations}

\author{S. Ya. Startsev\,\,\orcidlink{0000-0001-5891-6191}}\thanks{\href{http://www.researcherid.com/rid/D-1158-2009}{Web of Science ResearcherID: D-1158-2009}}

\address{Institute of Mathematics, Ufa Federal Research Centre, Russian Academy of Sciences}  

\begin{abstract} 
This paper deals with two kinds of B\"acklund transformations for scalar hyperbolic partial differential  equations. We prove that both these types of transformations map solutions of a Darboux integrable equation into solutions of, generally speaking, another but also Darboux integrable equation. The latter fact can be used to roughly check the completeness of a list of Darboux integrable equations. To illustrate this, we apply the above transformations to several equations from a well-known list of Darboux integrable equations and, as a result, obtain a Darboux integrable equation which is absent in this list, but is already known at present.
As a generalization of the last equation, we construct a family of Darboux integrable equations that is parametrized by three arbitrary functions, each of which depends on two arguments. This family is probably new.
\end{abstract}



\maketitle

\section{Introduction}
Among partial differential equations of the form 
\begin{equation}\label{hyp}
u_{xy}=F(x,y,u,u_x,u_y),
\end{equation}
an important class is formed by equations for which there exist functions
\begin{align*}
w(x,y,u,u_1, \dots, u_k), \quad &w_{u_k} \ne 0, &u_i:=\partial^i u /\partial x^i, \\ 
\bar{w}(x,y,u,\bar{u}_1, \dots, \bar{u}_m), \quad &\bar{w}_{\bar{u}_m} \ne 0, &\bar{u}_i:=\partial^i u /\partial y^i,
\end{align*}
such that for any solution of~\eqref{hyp} the function $w$ depends on $x$ only, and the function $\bar{w}$ depends on $y$ only. In other words, the relations $D_y(w)=0$ and $D_x(\bar{w})=0$ hold on solutions of~\eqref{hyp}, where $D_y$ and $D_x$ denote the total derivatives with respect to $y$ and $x$. In this case, the functions $w$ and $\bar{w}$ are respectively called an \emph{$x$-integral of order $k$} and a \emph{$y$-integral of order $m$} for the equation~\eqref{hyp}, and the equation itself is called \emph{Darboux integrable}. The most famous non-linear example of a Darboux integrable equation is the Liouville equation $u_{xy}=e^u$, for which $w= u_{xx} - u_x^2/2$, $\bar{w}= u_{yy} - u_y^2/2$.

It should be explained that, to check whether a relation holds on solutions of~\eqref{hyp}, we simply exclude from this relation all mixed derivatives of $u$ with respect to $x$ and $y$ by using the equation~\eqref{hyp} and its differential consequences (after this exclusion, the relation must be satisfied already identically). In particular, for any function $g$ of $x$, $y$, $u$, $u_i$, $\bar{u}_j$, we use the formulas 
\begin{eqnarray*}
&D_{x}(g) &= {\frac{\partial g}{\partial x}}+ {\frac{\partial
g}{\partial u}} u_1 + \sum^{\infty
}_{i=1}\left({\frac{\partial g}{\partial u_i}} u_{i+1}
+{\frac{\partial g}{\partial
\bar{u}_i}}D^{i-1}_{y}(F)\right), \\
&D_{y}(g) &=
\frac{\partial g}{\partial y}+\frac{\partial g}{\partial
u} \bar{u}_1 +\sum^{\infty }_{i=1}\left(\frac{\partial
g}{\partial \bar{u}_i} \bar{u}_{i+1} +\frac{\partial g}{\partial
u_i} D^{i-1}_x (F) \right)
\end{eqnarray*}
when these total derivatives are considered on solutions of~\eqref{hyp}. With this explanation in mind, it is easy to verify that, for example, $D_y(u_{xx} - u_x^2/2)=0$ and $D_x(u_{yy} - u_y^2/2)=0$ on the solutions of the above Liouville equation. 

The Darboux integrable equations have been studied for more than a hundred years, starting with classical works such as \cite{Gurs}. Interest in such equations continues to this day, as evidenced by relatively recent works. Among them, without pretending to be complete, we can mention \cite{ZhIzv}--\cite{ZhYu2}. In particular, the work \cite{ZhSok} contained a list of non-linear Darboux integrable equations~\eqref{hyp} which was declared complete (but with the caveat that the proof of this classification result is based on laborious calculations, takes about 200 pages, and there is a nonzero probability of mistakes). It was later noticed in \cite{Kap} that there is an error in this list, since it does not contain the Darboux integrable Laine equations \cite{Lai}. In this situation, it seems appropriate to use any, even rough, methods for checking the already known lists of Darboux integrable equations for completeness, especially if these methods are relatively simple and not laborious.

As such a relatively simple (and independent of any classification) way to test whether anything is missing in the known lists of Darboux integrable equations, in this paper we consider two kind of transformations: invertible differential substitutions introduced in \cite{Yam, SvSok}, and B\"acklund transformations inverse to the composition of the differential substitution $\tilde{u}=v_x$ and a point transformation $u=f(x,y,\tilde{u})$. At least one of the above types of transformations is applicable to almost any equation~\eqref{hyp} with the right-hand side linearly dependent on one of the derivatives. Therefore, relatively wide classes of equations~\eqref{hyp} admit these transformations, and within these classes there are Darboux integrable equations. In section~\ref{transf} we describe these transformations in more detail and prove that they map a Darboux integrable equation into, generally speaking, another but also Darboux integrable equation. Naturally, the closedness of a list of Darboux integrable equations with respect to the above transformations is a necessary but not sufficient condition for its completeness.

Section~\ref{exam} contains a list of Darboux integrable equations from~\cite{ZhSok} that most obviously allow the above transformations.\footnote{It cannot be ruled out that this list can be extended, for example, by some degenerate cases of equations from classes 6-8 in~\cite{ZhSok}.} As expected, the transformations for the overwhelming majority of the equations in this list lead again to the Darboux integrable equations contained in~\cite{ZhSok}.

But in one case, a transformation leads us to an equation that is missing in~\cite{ZhSok}. However, this equation is not new -- for example, it was mentioned as one of Moutard equations in~\cite{LaiM}. It was also shown in~\cite{LaiM} that a differential substitution relates it to one of the Darboux integrable Laine equations. The last fact allows us to construct a differential substitution that maps solutions of the Laine equation into solutions of an equation from~\cite{ZhSok}. The author does not know whether the Darboux integrability of the aforementioned Moutard equation was established back in Laine's time (strictly speaking, there is no corresponding statement in \cite{LaiM}), but at least in a much later work~\cite{ZhYu2} $x$- and $y$-integrals of this equation were in fact presented. As a generalization of the Moutard equation, we construct a family of Darboux integrable equations that is parametrized by three arbitrary functions, each of which depends on two arguments. This family is probably new.

To conclude the introduction, we note that, using the symmetry of the formula~\eqref{hyp} with respect to the interchange of $x$ and $y$, we further formulate and prove only one of the two ``symmetric'' propositions. All functions are assumed to be locally analytic in this article. The reasoning and statements are also local. In particular, when applying the implicit function theorem, we assume that the same branch of an implicit function is taken in all formulas containing this function. This does not prevent us from using the $\pm$ sign in the examples to enumerate all the branches of the implicit function (but each of them is still considered separately, without jumping from branch to branch inside one line of reasoning).

\section{ B\"acklund transformations}\label{transf}

\subsection{Reconstructing an original equation after the substitution $u=v_x$}\label{bclt}

It is easy to check that any equation of the form 
\begin{equation}\label{nou}
v_{xy}=G(x, y, v_x, v_y)
\end{equation} with the right-hand side depending on at least one of the derivatives (for definiteness, we assume that $G_{v_y} \ne 0$ ) admits a differential substitution $u=f(x,y,v_x)$, which transforms solutions of this equation into solutions of, generally speaking, another equation of the form~\eqref{hyp}. Indeed, for any function $f(x,y,v_x)$, $f_{v_x} \ne 0$, equation~\eqref{nou} can be written as
\begin{equation}\label{e1}
D_y(f(x,y,v_x))= g(x, y, v_x, v_y), \qquad g_{v_y} \ne 0.
\end{equation}
Denoting $f(x,y,v_x)$ as a new variable and taking \eqref{e1} into account, we get
\begin{equation}\label{ip}
u:=f(x,y,v_x), \qquad  u_y = g(x, y, v_x, v_y).
\end{equation}
Solving these equalities for $v_x$ and $v_y$, we obtain the system
\begin{equation}\label{e3}
v_x=\alpha (x,y,u), \qquad v_y =\varphi (x, y, u, u_y). 
\end{equation}
The compatibility condition of the system~\eqref{e3} gives us the conservation law
\begin{equation}\label{hcl}
D_x ( \varphi (x, y, u, u_y) ) = D_y (\alpha (x,y,u)), \qquad \varphi _{u_y} \alpha _u \ne 0,  
\end{equation}
which is essentially a new hyperbolic equation of the form~\eqref{hyp} (because \eqref{hcl} uniquely determines the corresponding $F$). Thus, the differential substitution $u=f(x,y,v_x)$ maps solutions of~\eqref{e1} into solutions of equation~\eqref{hcl}. Obviously, this substitution is the composition of the substitution $\tilde{u}=v_x$ and the point change of variables $u=f(x,y,\tilde{u})$.

A classical \cite{Gusr} particular case of the above transformation is the substitution $u=\sqrt{v_x}$, which maps solutions of the Goursat equation $v_{xy}=\sqrt{v_x v_y}$ into solutions of the equation $u_{xy}=u/ 4$.

In the present article, we are more interested not in the substitution~\eqref{ip}, but in the reverse transition from the equation \eqref{hcl} to the equation \eqref{e1}. This reverse transition is as follows. If equation~\eqref{hyp} admits a conservation law of the form~\eqref{hcl}, then for any solution of this equation we can define $v$ as a solution of the system~\eqref{e3} (this system is compatible due to~\eqref{hcl}). Expressing $u$ and $u_y$ in terms of $x$, $y$, $v_x$, $v_y$ by using~\eqref{e3}, we arrive at equalities of the form~\eqref{ip}, which are equivalent to the equation~\eqref{e1}.

For example, one can write the equation $u_{xy}=u/4$ in the form of the conservation law $D_x (u_y^2) = D_y (u^2/4)$ and define $v$ as a solution of the system $v_y=u_y^2$, $v_x =u^2/4$. The last equality gives us the expression $u = \pm 2 \sqrt{v_x}$. Differentiating this expression with respect to $y$ and comparing the result with the expression $u_y = \pm  \sqrt{v_y}$, we obtain the equation $v_{xy}=\pm \sqrt{v_x v_y}$.

Transformations based on conservation laws were considered, for example, in \cite{Muk}, where new independent variables $x$ and $y$ were introduced by using two special conservation laws of a hyperbolic equation. The article~\cite{GYL} also describes how a conservation law with a one-point conserved density of a differential-difference evolution equation generates a difference substitution that maps solutions of another differential-difference evolution equation into solutions of the first equation. In addition, the author actually used the transformation~\eqref{e3} in \cite{star} (without mentioning this fact in the article) to find the equation $v_{xy}= \sqrt{(1- 4 v_x^2) v_y }$, which is related to the sine-Gordon equation by a substitution of the form~\eqref{ip}.

\subsection{Invertible differential substitution}\label{invds}
It is easy to see that the equation~\eqref{hcl} is a particular case of the equation of the form
\begin{equation}\label{cl}
D_x ( \varphi (x, y, u, u_y) ) = \psi (x,y,u,u_y), \qquad \varphi _{u_y} \psi _u - \varphi _{u} \psi_{u_y} \ne 0.  
\end{equation}
For equations ~\eqref{hyp} writeable in the form~\eqref{cl}, the works \cite{Yam,SvSok} considered differential substitutions based on the same scheme as the transformation described above: introduce a new variable $v$,  express $u$ and/or its derivatives in terms of $v$ and its derivatives by virtue of~\eqref{hyp}, and then obtain a new hyperbolic equation from a compatibility condition of these expressions. As applied to~\eqref{cl}, this scheme looks as follows.

Let us introduce the new variable
\begin{equation}\label{fp}
v=\varphi (x, y, u, u_y).
\end{equation}
Then~\eqref{cl} is equivalent to the system
\begin{equation}\label{it}
v=\varphi (x, y, u, u_y), \qquad  v_x = \psi (x,y,u,u_y).
\end{equation}
Solving \eqref{it} for $u$, $u_y$, we obtain a system of the form	
\begin{equation}\label{pqs}
u=p(x,y,v,v_x), \qquad u_y = q (x, y, v, v_x). 
\end{equation}
Differentiating the first equation of this system with respect to $y$ and comparing the result with the second equation, we arrive at the equation 
\begin{equation}\label{pq}
D_y ( p(x,y,v,v_x) ) = q (x, y, v, v_x), 
\end{equation}
which is equivalent to \eqref{pqs}. Repeating the above reasoning in reverse order, it is easy to check that the differential substitution $u=p(x,y,v,v_x)$ maps solutions of~\eqref{pq} into solutions of~\eqref{cl} (for more details, see the proof of Theorem~\ref{t1}).

According to~\cite{SvSok}, any equation
\[ u_{xy} = b(x,y,u,u_y) u_x + c(x,y,u,u_y), \qquad c_u + c_{u_y} b \ne b_x + b_{u_y} c \]
can be written in the form~\eqref{cl}, choosing as $\varphi$ a solution of the equation $b \varphi_{u_y} + \varphi_{u}=0$. Any other solution of the last equation has the form $\gamma(x,y,\varphi)$ and this arbitrariness in the choice of $\varphi$ corresponds to the point change of variables $\tilde{v}=\gamma(x,y,v)$ in~\eqref{pq}. In some cases, it is possible to choose $\gamma$ so that~\eqref{cl} takes the form~\eqref{hcl}. For example, the equation $u_{xy} = u_x /(x+y)$ can be written in the form~\eqref{cl} as follows
\begin{equation}\label{prex}
D_x \left( u_y - \frac{u}{x+y} \right) = \frac{u}{(x+y)^2}.
\end{equation}
Choosing $\gamma=\varphi^2$ and taking~\eqref{prex} into account, it is easy to check that
\begin{equation*}
D_x \left( \left( u_y - \frac{u}{x+y} \right)^2 \right) = 2 \left( u_y - \frac{u}{x+y} \right) \frac{u}{(x+y)^2} = D_y \left(\frac{u^2}{(x+y)^2} \right).
\end{equation*}

Thus, the equations~\eqref{hyp} representable in the form~\eqref{hcl} can be related both to equations~\eqref{e1} via the B\"acklund transformation~\eqref{e3} and to equations~\eqref{pq} by the differential substitution~\eqref{fp}. It is clear that the equations~\eqref{e1} and \eqref{pq} are, generally speaking, different. For example, as shown above, the equation $u_{xy}=u/4$ is related by a B\"acklund transformation~\eqref{e3} to the Goursat equation, while the substitution $v=u_y$ of the form~\eqref{fp} maps solutions of this linear\footnote{Up to a point change of variables, the transformation~\eqref{fp} for linear equations coincides with the local Laplace transformation; you can read more about the latter transformation, for example, in \cite{Tr}, \cite{ZhSok}.} equations again into solutions of the same equation $v_{xy}=v/4$.

\subsection{Preservation of the Darboux integrability under the transformations}

We will prove below that both the substitutions~\eqref{fp} and the transformations~\eqref{e3} preserve the Darboux integrability. For this proof we need the following statement.

\begin{lemma}\label{sint}
Let equation~\eqref{hyp} can be written in the form~\eqref{cl} and the relation
\begin{equation}\label{pint}
D_x (\bar{w}(x,y,u,\bar{u}_1, \dots, \bar{u}_m)) = h(x,y,u,\bar{u}_1, \dots, \bar{u}_m)
\end{equation}
hold on solutions of this equation for some functions $\bar{w}$ and $h$. Then $\bar{w}$ is a function of $x$, $y$,  $\varphi(x, y, u, u_y)$ and its total derivatives with respect to $y$:
\begin{equation}\label{ze}
\bar{w}= \bar{\theta} (x,y, \varphi, D_y(\varphi), \dots, D_y^{m-1} (\varphi)).
\end{equation}
\end{lemma}
It should be noted that in this paper we only need the case $h=0$.

\begin{proof} Let us rewrite $\bar{w}$ as $\bar{\theta} (x,y, u, \varphi, D_y(\varphi), \dots, D_y^{m-1} (\varphi))$. The equalities $D_x(D_y^i(\varphi)) = D_y^i(\psi)$ hold on solutions of~\eqref{hyp} as a consequence of~\eqref{cl}. Therefore,~\eqref{pint} becomes the form
\begin{equation}\label{loc}
\bar{\theta}_x + u_x \bar{\theta}_u + \psi \bar{\theta}_{\varphi} + \dots = h. 
\end{equation}
The above equality does not contain mixed derivatives of $u$ with respect to $x$ and $y$. That is, this equality already takes account of all dependencies between derivatives generated by the fact that~\eqref{pint} holds on solutions of~\eqref{hyp}. Therefore,~\eqref{loc} must hold identically. Since only the second term on the left-hand side depends on $u_x$ in~\eqref{loc}, the relation \eqref{pint} can hold only in the case $\bar{\theta}_u = 0$.
\end{proof}

\begin{theorem}\label{t1} Let equation~\eqref{hyp} can be written in the form~\eqref{cl} and admit an $x$-integral $w$ of order $k$ (a $y$-integral $\bar{w}$ of order $m$). Then the equation~\eqref{pq} related to~\eqref{cl} by the transformation \eqref{fp}-\eqref{pqs} admits the $x$-integral
\begin{equation}\label{pqy}
\theta = w(x,y,p,D_x(p),\dots, D_x^k(p)) 
\end{equation}
of order $k+1$ (the $y$-integral $\bar{\theta}(x,y,v,\dots,\bar{v}_{m-1})$ of order $m-1$, where $\bar{\theta}$ is related to $\bar{w}$ by formula~\eqref{ze}). 
\end{theorem}
Thus, the above theorem can be applied in situations where the equation~\eqref{cl} has an integral with respect to only one of the characteristics, but guarantees the Darboux integrability of the equation~\eqref{pq} if the equation~\eqref{cl} is Darboux integrable. Applying exactly the same (up to interchanging $x$ and $y$) statement to the reverse transition from~\eqref{pq} to~\eqref{cl}, we get that the presence of $x$- and $y$-integrals of order $k+1$ and $m-1$ for~\eqref{pq} implies the existence of $x$- and $y$-integrals of order $k$ and $m$ for the equation~\eqref{cl}. Therefore, if $w$ and $\bar{w}$ are integrals of the smallest orders for~\eqref{cl}, then $\theta$ and $\bar{\theta}$ are also integrals of the smallest orders for~\eqref{pq}. Note that the case of first-order $y$-integrals for~\eqref{cl} is impossible by Lemma~\ref{sint} (otherwise $\varphi$ and $\psi$ turn out to be functions of $x$, $y$ and $\bar{w}$ that contradicts the inequality in the formula~\eqref{cl}).
\begin{proof}
Let us consider in more detail the reverse transition from the equation~\eqref{pq} to the equation~\eqref{cl}. Starting from~\eqref{pq} and introducing a ``new'' variable $u=p(x,y,v,v_x)$, we obtain the system~\eqref{pqs} by virtue of the equation~\eqref{pq}. The functions $p$ and $q$ satisfy the identities
\[ v \equiv \varphi (x, y, p, q), \qquad  v_x \equiv \psi (x,y,p,q) \]
by construction. Therefore, the first of formulas~\eqref{pqs} maps solutions of~\eqref{pq} into functions $u$ that satisfy the equalities~\eqref{it} and, as a consequence, the equation~\eqref{cl}. Thus, 
\begin{equation}\label{up}
u=p(x,y,v,v_x) 
\end{equation}
is a solution of~\eqref{cl} for any solution of~\eqref{pq}. And since $D_y(w)=0$ on solutions of~\eqref{cl} (in particular, on solutions obtained from  solutions of~\eqref{pq} by the differential substitution~\eqref{up}), the equality $D_y(\theta)=0$ holds on solutions of~\eqref{pq}, where $\theta$ is given by formula~\eqref{pqy}.

Similarly, since the first of the equalities~\eqref{it} holds for any solution $v$ of~\eqref{pq} and the corresponding solution of~\eqref{cl} obtained from $v$ by formula~\eqref{up}, we have 
\begin{align*}
D_x(\bar{\theta}(x,y,v,\dots,\bar{v}_{m-1})) &= D_x(\bar{\theta}(x,y,\varphi (x,y,u,u_y),\dots,D_y^{m-1}(\varphi (x,y,u,u_y)))) \\
& = D_x (\bar{w}(x,y,u,\dots,\bar{u}_m)) =0
\end{align*} 
on solutions of~\eqref{pq}. 
\end{proof}

The existence of integrals and, in particular, the Darboux integrability is also preserved in going from the equation~\eqref{hcl} to the equation~\eqref{e1} by using the B\"acklund transformation~\eqref{e3}. This fact is proved almost completely in the same way as in the proof of Theorem~\ref{t1}. This is why we omit the proof of the following proposition.

\begin{theorem}\label{t2} Let equation~\eqref{hyp} can be written in the form~\eqref{hcl} and admit an $x$-integral $w$ of order $k$ (a $y$-integral $\bar{w}$ of order $m$). Then the equation~\eqref{e1} related to~\eqref{hcl} by the B\"acklund transformation~\eqref{e3} admits the $x$-integral
\begin{equation*}
\theta = w(x,y,f,D_x(f),\dots, D_x^k(f)) 
\end{equation*}
of order $k+1$ (the $y$-integral $\bar{\theta}(x,y,v_y,\dots,\bar{v}_{m})$ of order $m$, where $\bar{\theta}$ is related to $\bar{w}$ by formula~\eqref{ze}). 
\end{theorem}
Note that, in contrast to Theorem~\ref{t1}, the order of the integral $\bar{\theta}$ is not less than the order of the integral $\bar{w}$ in Theorem~\ref{t2}. However, there are no guarantees that $\theta$ and $\bar{\theta}$ obtained by the formulas of Theorem~\ref{t2} are the smallest order integrals for~\eqref{e1} even if corresponding $w$ and $\bar{w}$ are integrals of the smallest orders for~\eqref{hcl}.  Examples show that the smallest orders of integrals for the equation~\eqref{e1} are often  less than the orders of the integrals that can be obtained by using Theorem~\ref{t2}.

\section{Examples of the transformations \protect\\ for Darboux integrable equations}\label{exam}

\subsection{The list of tested examples}

Among the Darboux integrable equations~\eqref{hyp} in the paper~\cite{ZhSok}, at least equations with the following right-hand sides 
\begin{equation}\label{spi}
e^u,\quad u u_y,\quad e^u u_y,\quad (e^u+e^{-u}) u_y,\quad \sqrt{1-u_y^2},\quad e^u \sqrt{1-u_y^2}   
\end{equation}
can be written in the form~\eqref{cl}. We can also represent all the equations from the list~\eqref{spi} in the form~\eqref{hcl}.
The equation
\begin{equation}\label{gib}
u_{xy} = \left( \frac{u_y}{u-x} + \frac{u_y}{u-y} \right) u_x 
\end{equation}
can be written in the form~\eqref{cl} too.

As expected, the B\"acklund transformations~\eqref{e3} and the differential substitutions~\eqref{fp} in most cases (namely, for all the equations listed in~\eqref{spi}) does not take us beyond the list of Darboux integrable equations from the paper~\cite {ZhSok}. Among examples of this kind, we consider only the equation
\begin{equation}\label{root}
u_{xy} = \sqrt{1-u_y^2}. 
\end{equation}
This equation is remarkable due to the fact that the B\"acklund transformation~\eqref{e3} (as well as the inverse substitution~\eqref{ip}) maps~\eqref{root} into itself.

Indeed, \eqref{root} can be written as $D_x \left( \sqrt{1-u_y^2} \right) = - D_y (u)$. Then
\[ v_x= u, \quad v_y= -\sqrt{1-u_y^2} \quad \Rightarrow \quad u_y= \pm \sqrt{1-v_y^2}. \]
Differentiating the first of the above equalities with respect to $y$ and comparing the result with the third one, we arrive at the equation $v_{xy} = \pm \sqrt{1-v_y^2}$, which coincides with~\eqref{root} up to the change of notation $u=\pm v$. Quite similarly to~\cite{S21}, one can prove that the substitution~\eqref{fp} cannot map a Darboux-integrable equation of the form \eqref{cl} into itself. The just given example shows that the same is not true for the transformations~\eqref{ip}-\eqref{e3}.

Applying a differential substitution of the form~\eqref{fp} to the equation~\eqref{gib}, we obtain a Darboux integrable equation, which this time is absent in the list from~\cite{ZhSok}. To do this, we write the equation~\eqref{gib} as
\[ D_x \left( \ln \left( \frac{u_y}{(u-x)(u-y)} \right) \right) = \frac{1}{u-x}. \]
Following the scheme of the transformation described in subsection~\ref{invds}
we obtain
\begin{align}
e^v &= u_y/((u-x)(u-y)), &v_x = 1/(u-x); \label{1g} \\
u_y &= e^v  v_x^{-1} (x-y + v_x^{-1}) ,  &u = x + v_x^{-1}. \label{2g}
\end{align}
The compatibility condition of the system~\eqref{2g} leads to the equation
\begin{equation}\label{mut}
 v_{xy} = (y-x) e^v v_x - e^v.
\end{equation} 
It is easy to see from the equation~\eqref{mut} itself that $\bar{\theta}=v_y + (x-y) e^v$ is a first-order $y$-integral of this equation, and its $x$-integral
\begin{equation}\label{mutyi}
\theta = D_x (\ln (v_{xx} - v_x^2)) - v_x 
\end{equation}
of the smallest order is obtained from the smallest order $x$-integral 
\[w = u_{xx}/u_x + (1 -2 u_x)/(u-x )\]
of the equation~\eqref{gib} by using formula~\eqref{pqy} with $p=x + v_x^{-1}$.

Below we prove that the equation~\eqref{mut} is absent in~\cite{ZhSok}. However, this equation is not new. For example, it was given in~\cite{LaiM} as one of Moutard equations\footnote{In~\cite{LaiM}, an equations of the form~\eqref{hyp} was called a Moutard equation if there exists an explicit formula for its solutions that depends on arbitrary functions $X(x)$, $Y(y)$ and their derivatives. In the same work $v=\ln \frac{Y''}{Y-Y'(y-x)+X}$ was indicated as such a formula for~\eqref{mut}. The existence of solutions of this kind is used as an alternative definition of Darboux integrability in some papers. But, as far as the author knows, there is no rigorously proven statement about the existence of $x$- and $y$-integrals for such equations.}. The author does not know whether the Darboux integrability of~\eqref{mut} was established back in Laine's time (strictly speaking, there is no corresponding statement in \cite{LaiM}), but at least in a much later work~\cite{ZhYu2} $x$- and $y$-integrals of this equation were in fact presented. To see this, it suffices to compare the last paragraph of Section 2 in~\cite{ZhYu2} with the formula for $W$ immediately after formula~(58) in the same paper.

By the way, in the paper~\cite{LaiM} mentioned above, Laine related~\eqref{mut} with one of the Darboux integrable equations he found, namely with the equation
\begin{equation}\label{lai}
\vartheta_{xy} = \vartheta_y \left( \frac{\vartheta_x}{\vartheta-x}+\frac{\vartheta_x}{\vartheta-y} \right) + \frac{\vartheta_y \sqrt{\vartheta_x}}{\vartheta-x} ,
\end{equation} 
by using the differential substitution $2e^v =\vartheta_y/((\vartheta-x)(\vartheta-y))$, which almost coincides with the first formula~\eqref{1g}. Let us consider the composition of this differential substitution and the substitution $u=x + v_x^{-1}$, which maps solutions of~\eqref{mut} into solutions of~\eqref{gib}. If we eliminate mixed derivatives of $\vartheta$ from this composition by using~\eqref{lai}, then we get that the substitution $u= x+(\vartheta-x)/(\sqrt{\vartheta_x}+1)$ maps solutions of~\eqref{lai} into solutions of the equation~\eqref{gib}.

\subsection{A note on the equivalence problem and a generalization of the Moutard equation}
The work~\cite{ZhSok} enumerates Darboux integrable equations up to point transformations
\begin{equation}\label{ptran}
x=\xi (\tilde{x}), \qquad y=\eta (\tilde{y}), \qquad u=\lambda(\tilde{x},\tilde{y},\tilde{u})
\end{equation}
and the interchange $x \leftrightarrow y$.
To demonstrate that the Moutard equation~\eqref{mut} cannot be reduced to an equation from~\cite{ZhSok} by a transformation of the form~\eqref{ptran}, it is convenient to employ functions $g$ of $x$, $y$, $u$ and its derivatives with respect to $x$ such that the relation
\begin{equation}\label{class}
F_{u_x} = D_y (g(x,y,u,u_1,\dots, u_n))
\end{equation}
holds on solutions of~\eqref{hyp}. The number $n$ is called the order of the function $g$ if $g_{u_n} \ne 0$, and we define the order of $g$ to be zero when $g$ does not depend on derivatives of $u$. If equation~\eqref{hyp} possesses an $x$-integral $w$ of order $k$, then $g=-\ln(w_{u_k})$ satisfies the relation~\eqref{class} (see~\cite{ZhIzv} for example). This is why such functions $g$ always exist for any Darboux integrable equation~\eqref{hyp}. 

Let $r$ be the smallest order of $x$-integrals for equation~\eqref{hyp}, and let a function $g$ of order $n<r$ satisfy the relation~\eqref{class} on solutions of this equation. Then there is no function $\tilde{g}$ of order $\tilde{n}<n$ satisfying~\eqref{class}, because $g-\tilde{g}$ is an $x$-integral of order $n< r$ otherwise. In particular, the last fact together with the formulas~\eqref{mutyi} and $g=-\ln(\theta_{v_{xxx}})$ imply that the relation~\eqref{class} cannot hold on solutions of~\eqref{mut} for a function $g$ of order less than 2.

On the other hand, the straightforward calculation shows that the relation~\eqref{class} and the order of the corresponding function $g$ are preserved under transformation~\eqref{ptran}. Hence, the Moutard equation~\eqref{mut} cannot be related via~\eqref{ptran} to any equation for which the relation~\eqref{class} holds with $g$ of order 1 or 0. But for almost all Darboux integrable equation in~\cite{ZhSok} both the relation~\eqref{class} and its ``symmetric'' version 
\[ F_{u_y} = D_x (\bar{g}(x,y,u,\bar{u}_1,\dots, \bar{u}_\ell )) \]
hold for functions $g$ and $\bar{g}$ that do not depend on the second and higher derivatives of $u$. 
After excluding these equations from consideration, the only Darboux integrable equation remains in the work \cite{ZhSok}. But the last equation (numbered by 9 in~\cite{ZhSok}) has no first-order integral and therefore also cannot be related to~\eqref{mut} by transformations~\eqref{ptran}. Thus, the equation~\eqref{mut} is missed in~\cite{ZhSok}.

The Moutard equation~\eqref{mut} can be generalized in the following way. Any function $\bar{w}(x,y,u,u_y)$, $\bar{w}_{u_y} \ne 0$, is a $y$-integral of the equation
\begin{equation}\label{Pint} 
u_{xy} = - \frac{\bar{w}_u u_x + \bar{w}_x}{\bar{w}_{u_y}}.   
\end{equation}
Let $\bar{w}$ be defined by an equality of the form $u_y=\zeta (x,y,u,\bar{w})$, $\zeta_{\bar{w}} \ne 0$. Differentiating this equality with respect to $u_y$, $u$ and $x$, we respectively obtain 
\[ \bar{w}_{u_y} = \frac{1}{\zeta_{\bar{w}}}, \qquad  \bar{w}_u = - \frac{\zeta_u}{\zeta_{\bar{w}}}, \qquad  \bar{w}_x = - \frac{\zeta_x}{\zeta_{\bar{w}}}. \]
Therefore, the equation~\eqref{Pint} can be rewritten as
\begin{equation}\label{zint} 
u_{xy} = \zeta_u u_x + \zeta_x.
\end{equation}
Taking $D_x(\bar{w})=0$ into account and assuming  $\zeta$ has the form $(x\, \alpha(y,\bar{w}) + \beta(y,\bar{w})) e^u + \gamma(y,\bar{w})$, we see that
\begin{equation}\label{prei} 
D_y (\ln (u_{xx}-u_x^2))=\zeta_u, \qquad D_x(\zeta_u) = \zeta_u u_x + \zeta_x 
\end{equation}
on solutions of~\eqref{zint}. Formulas~\eqref{prei} imply that $w=D_x(\ln (u_{xx}-u_x^2)) - u_x$ is an $x$-integral of~\eqref{zint}. Indeed,
\[ D_y (w) = D_x (D_y(\ln (u_{xx}-u_x^2))) - \zeta_u u_x - \zeta_x = D_x(\zeta_u) - \zeta_u u_x - \zeta_x = 0. \]

Thus, the equation~\eqref{Pint} is Darboux integrable if there exist functions $\alpha$, $\beta$ and $\gamma$ such that the function $\bar{w}$ satisfies the relation
\begin{equation}\label{dco}
u_y = (x\, \alpha(y,\bar{w}) + \beta(y,\bar{w})) e^u + \gamma(y,\bar{w}).
\end{equation}
The Moutard equation~\eqref{mut} is a particular case of~\eqref{Pint} with $\alpha=-1$, $\beta=y$ and $\gamma=\bar{w}$. Therefore, the equation~\eqref{Pint} with $\bar{w}$ defined by~\eqref{dco} is, generally speaking, absent in~\cite{ZhSok} and is probably new as a Darboux integrable equation in the case $\alpha \ne 0$. (The work~\cite{ZhSok} contains Darboux integrable equations of the form~\eqref{Pint}, but only in the case $\bar{w}_u  \bar{w}_x = 0$.)

\end{document}